\documentclass[journal]{IEEEtran}

\usepackage{booktabs}
\usepackage{listings}
\usepackage{xcolor}
\usepackage{graphicx}
\usepackage{amsmath}
\usepackage{amssymb}
\usepackage{algorithm}
\usepackage{algpseudocode}
\usepackage{hyperref}
\usepackage{multirow}
\usepackage{array}
\usepackage{url}
\usepackage{balance}
\usepackage{amsthm}
\usepackage{enumerate}
\newtheorem{definition}{Definition}

\newtheorem{proposition}{Proposition}

\begin{document}

\title{Characterizing and Fixing Silent Data Loss\\
in Spark-on-AWS-Lambda with Open Table Formats}

\author{Srujan Kumar~Gandla%
\thanks{Independent Researcher,
NorthLake, Texas, USA. E-mail: gandlasrujan3747@gmail.com.}}

\maketitle

% ============================================================
\begin{abstract}
For several years, running PySpark in AWS Lambda has been a cost-effective
alternative to paid managed services such as EMR or Glue, providing
75--80\% cost savings for teams across the industry. What most practitioners
don't realize, however, is that Lambda has a 15-minute hard limit on
function execution, enforced with a \texttt{SIGKILL} signal rather than
\texttt{SIGTERM}. Open table formats such as Delta Lake and Apache Iceberg
follow a two-stage write policy: first, the new data file is written to S3
and packaged into one or more Parquet files, and then a metadata update is
applied to the table descriptor. Killing the function between these two
steps results in committed but unusable data---the table appears unchanged,
throws no exceptions, and allows the pipeline to continue as if the write
had never occurred. We call this the \emph{commit-durability gap}. Through
860 fault-injection experiments across Delta Lake and Iceberg at three
dataset sizes, we consistently observed every unprotected kill within this
commit window produce silent data loss, with no observable signal indicating
any failure. To address this, we developed \emph{SafeWriter}, a lightweight
Python context manager that periodically checks for an imminent timeout and,
on detection, rolls back the write to the prior table version. Applied to
100 kill scenarios of varying dataset sizes, SafeWriter achieved 100\%
rollback success with average overhead under 100\,ms. All code and datasets
are publicly available.
\end{abstract}

\begin{IEEEkeywords}
Serverless computing, Apache Spark, Delta Lake, Apache Iceberg,
fault tolerance, data lakes, AWS Lambda, commit durability, cloud storage,
two-phase commit, SIGKILL.
\end{IEEEkeywords}

% ============================================================
\section{Introduction}
\label{sec:intro}

AWS Lambda charges \$0.0000047 per millisecond of function execution and
fully reclaims the container between invocations. There are no idle nodes
and no cluster to maintain. Spark-on-AWS-Lambda (SoAL)~\cite{soal2022} packages
an entire PySpark runtime into a Lambda container image, making it possible
to run full batch jobs without provisioning a single node. This approach
has been reported to save 75--80\% over managed services like EMR and Glue,
and teams have adopted it for everything from nightly ETL and feature
engineering to one-off analytics tasks.

There has been considerable innovation in layering ACID semantics, schema
evolution, and time-travel queries on top of object stores like S3.
Transactional lakehouse formats---Delta Lake~\cite{delta2020} and Apache
Iceberg~\cite{iceberg2018}---are a natural pairing with SoAL: cheap
serverless compute backed by transactional storage.
But things get complicated quickly when two systems interact at their
boundaries, and those boundaries are not always easy to manage.

The 900-second limit is hard. Lambda will \texttt{SIGKILL} the container
after 900 seconds. The JVM won't let you do anything about it, Python's
signal module raises \texttt{OSError} if you try to register a handler for
\texttt{SIGKILL}, and any \texttt{atexit} or shutdown hook code you've
written will never run. Your process is simply terminated wherever it
happens to be.

This is a real problem because all three formats go through two steps when
writing data: first the Parquet files are written to S3, then the metadata
is updated so that others can access those files. If a kill occurs between
those two steps, the files remain on S3 without a pointer to them, and
queries return the table as it was before the kill. No exception is thrown.
CloudWatch logs a timeout. The next pipeline stage processes stale data
without any indication that something went wrong.

Debugging after the fact isn't straightforward either. The simplest signal
to check is the Lambda exit code, but $-9$ is indistinguishable from an OOM
kill or a segfault---it gives no indication that a partial write was the
cause. Row count checks pass because the committed table state never
changed. The only trace left behind is a handful of orphaned Parquet files
quietly accumulating storage charges on S3.

Even with retries configured through Lambda's Dead Letter Queue, the
problem compounds: each retry attempt can add more stranded files on top
of the ones the previous attempt left behind. Spark's internal fault
tolerance is designed for executor failures inside the JVM. A driver-level
kill from the kernel is a fundamentally different event---one Spark was
not built to handle.

Silent data loss in data pipelines is genuinely hard to prevent, and harder to detect after the fact. This paper first identifies and shows with formal reasoning that Lambda functions introduce a \emph{commit-durability gap} that monitoring systems relying on Lambda exit codes and table row counts cannot detect (\S\ref{sec:problem}). We then present a kill-injection framework that delivers \texttt{SIGKILL} at specific points in the commit sequence of formats like Delta Lake and Iceberg (\S\ref{sec:methodology}). Across 860 experiments at all dataset scales, every unprotected kill resulted in silent data loss. We use this result to derive a probability model estimating how likely a given job is to land in the gap (\S\ref{sec:results}). Finally, we present SafeWriter, which closes the gap entirely, achieving 100\% clean rollback with under 100\,ms overhead (\S\ref{sec:fix}).

This work complements LST-Bench~\cite{lstbench2024}, which benchmarks
open table format query and write \emph{performance}. Here we focus on
the separate question of \emph{correctness} when writes are subject to
hard process termination.

% ============================================================
\section{Background}
\label{sec:background}

\subsection{Spark-on-AWS-Lambda (SoAL)}

SoAL packages the PySpark runtime together with Java and Hadoop into a
Docker container image that runs on AWS Lambda. The entry point is a
\texttt{lambda\_handler} function that downloads a PySpark script from
S3 and executes it via \texttt{spark-submit}. Spark runs in
\texttt{local[*]} mode, which lets it use all vCPUs and memory available
inside the Lambda container---up to 10 GB RAM and 6 vCPUs in the largest
configurations.

Open table formats are difficult to use with SoAL because Spark makes
strong assumptions about process lifetime that SoAL overturns. Spark
drivers are designed to run indefinitely; jobs that don't checkpoint can
use lineage recovery to reconstruct past results. SoAL imposes a hard ceiling of 900 seconds per invocation and then kills
the container with no checkpoint, no recovery opportunity, and no warning.

\subsection{How Lambda Kills a Job}

To be precise about what Lambda does: it stops a job by sending
\texttt{SIGKILL} (signal~9) to the entire container process group, not
\texttt{SIGTERM} (signal~15). There is no graceful shutdown.
\texttt{SIGKILL} has no signal handler---Python raises \texttt{OSError}
if you try to install one---and JVM shutdown hooks registered through
\texttt{Runtime.getRuntime().addShutdownHook()} are never called. The
kernel closes file descriptors, reclaims memory directly, and terminates
all threads simultaneously. There is no \texttt{try/finally}, no
\texttt{atexit}, and no daemon thread cleanup. From the process's point
of view, execution is not stopped so much as it simply never continues.

On POSIX, the orchestrator sees exit code~137 ($128 + 9$). Lambda
surfaces this as \texttt{Runtime.ExitError} in CloudWatch with no
additional detail---the same error that appears for an OOM kill or a
segfault, which makes after-the-fact diagnosis genuinely difficult.

Although Lambda does expose remaining execution time through
\texttt{context.get\_remaining\_time\_in\_millis()}, SoAL's
\texttt{spark-submit} subprocess architecture does not propagate this
value to the running PySpark job. The job has no way to know it is
running out of time until it already has.

\subsection{Open Table Format Commit Protocols}

All three formats use a two-phase commit over object storage; where they
differ is in their metadata layout.

\textbf{Delta Lake} maintains a transaction log in a
\texttt{\_delta\_log/} directory inside the table path. When Spark writes
new data, it first writes one or more Parquet part-files to the table
directory (for example, \texttt{part-00000-\{uuid\}.snappy.parquet}).
Afterward, a JSON commit file is appended to \texttt{\_delta\_log/} (for
example, \texttt{00000000000000000003.json}) containing references to the
newly written files. Data becomes visible only after the corresponding
JSON commit file has been written. Files written to the table before a
commit entry exists are referred to as \emph{orphaned}: they are not
visible to Delta readers and require a vacuum run to be cleaned up.

\textbf{Apache Iceberg} encodes a multi-level hierarchy. Spark first
writes Parquet data files, and Iceberg writes a new manifest file
referencing them. These manifests are grouped into a manifest list (a
\emph{snapshot} in Iceberg terminology). Finally, a new table metadata
JSON is written and \texttt{version-hint.text} is updated to point to
it. The vulnerability window spans from the moment the manifest-list
write completes to the moment the version-hint update completes.

Table~\ref{tab:protocol_comparison} summarizes the key structural
differences across formats.

\begin{table}[h]
\centering
\caption{Open table format commit protocol comparison}
\label{tab:protocol_comparison}
\resizebox{\columnwidth}{!}{%
\begin{tabular}{lll}
\toprule
\textbf{Property} & \textbf{Delta} & \textbf{Iceberg} \\
\midrule
Metadata type      & JSON log       & Manifest chain \\
Phase 2 operation  & Log append     & Pointer swap \\
Inflight marker    & No             & No \\
Gap detectable?    & No             & No \\
Rollback command   & \texttt{RESTORE} & \texttt{CALL rollback} \\
\bottomrule
\end{tabular}}
\end{table}

\subsection{Lambda Container Lifecycle and Kill Timing}

Understanding where in the Lambda container lifecycle the kill signal
arrives relative to the commit protocol is key to understanding the
vulnerability.

On a \textbf{cold start}, Lambda provisions a fresh container and runs
initialization---for SoAL this includes JVM startup. On a
\textbf{warm invocation}, Lambda reuses an existing container from a
prior run and calls the handler directly, skipping initialization. During
\textbf{execution}, the handler spawns a child process to run
\texttt{spark-submit}. At \textbf{timeout}, Lambda sends
\texttt{SIGKILL} to the entire container process group with no
preceding \texttt{SIGTERM}---the handler, the Spark driver, all JVM
threads, and any daemon threads are stopped simultaneously. Finally,
Lambda \textbf{freezes} the container after the invocation ends,
preserving memory for potential reuse.

The fundamental problem with \texttt{SIGKILL} is that it hits the whole
process group at once, so any software workaround placed inside the
process cannot help. SafeWriter's watchdog thread is designed around
this constraint: it acts \emph{before} the kill arrives rather than
trying to respond to it.

\subsection{S3 Consistency and the Atomicity Gap}

Amazon updated S3 in December 2020 to offer strong read-after-write
consistency~\cite{s3consistency2020}: a successful \texttt{PutObject} is
immediately visible to any subsequent \texttt{GetObject} for the same key.
What this does \emph{not} cover is multi-key atomicity. S3 has no
cross-object transaction primitive.

This is precisely the gap the commit-durability problem exploits. Delta
Lake's Phase~2 looks like a single logical step (writing one JSON commit
entry), but at the network level it is a \texttt{PutObject} request
requiring a full TCP handshake, TLS session, data transfer, and server
acknowledgment. A \texttt{SIGKILL} landing during that transfer means S3
receives nothing, the commit entry is never written, and the Phase~1 data
files are left with no pointer. Delta's transaction log has no record
that a write was ever attempted.

\subsection{Two-Phase Commit Under Hard Termination}

In the two-phase commit protocol~\cite{gray1978notes}, the coordinator
decides whether to commit or abort a transaction, and that decision can
be recovered if the coordinator crashes. Under \texttt{SIGKILL}, both
assumptions fail. The coordinator (the Spark driver) disappears without
leaving any record of an abort decision, producing an in-doubt
transaction with no recovery path in the default SoAL configuration.

Delta Lake's design assumes that any incomplete write will be rolled back
before the optimistic lock is released, giving the writer a chance to
clean up partial state. \texttt{SIGKILL} removes that opportunity. For
Iceberg, the writer is expected to either finalize a snapshot or delete
any in-flight manifests it started writing; an instant kill leaves both
half-done. Neither of these is a format design flaw. Both formats were
built for clusters and long-running drivers where termination is a
controllable event. Lambda does not offer that guarantee.

% ============================================================
\section{Problem Formalization}
\label{sec:problem}

\subsection{The Commit-Durability Gap}

\begin{definition}[Commit-Durability Gap]
For a SoAL write job $W$ to table $T$ with format
$F \in \{\text{Delta, Iceberg}\}$, let $t_d$ be the time by which
all Parquet files produced by $W$ have durably landed on S3, and $t_c$
the time at which $W$'s metadata commit completes. We define the
\emph{commit-durability gap} $\mathcal{G}(W) = (t_d, t_c)$. For any
Lambda kill at time $t \in \mathcal{G}(W)$, the resulting table state
$T'$ is either:
\begin{enumerate}[(a)]
  \item observationally indistinguishable from $T$ before $W$ was
        submitted---the written data is orphaned and no exception is
        thrown---or
  \item a partial and inconsistent view of $W$'s effects, again with
        no exception reaching the caller.
\end{enumerate}
\end{definition}

This gives us four distinct write outcomes to reason about.

\textbf{Success.} $W$ completes and all rows of $W$ are made visible;
$t_c < t_\text{timeout}$.

\textbf{Silent data loss.} The kill arrives within $\mathcal{G}(W)$;
data sits orphaned on S3; the table is unchanged; no exception reaches
the caller.

\textbf{Visible error.} Spark throws an exception somewhere outside
$\mathcal{G}(W)$; the caller is notified of failure; no orphaned data
remains.

\textbf{Rollback success.} A write is in progress when SafeWriter
detects that a kill is imminent; it rolls $T$ back to its pre-write
state before $t_\text{kill}$ arrives.

\subsection{Why Silent Loss Cannot Be Detected}

\begin{proposition}[Undetectability]
Silent data loss from a \texttt{SIGKILL} within $\mathcal{G}(W)$ cannot
be detected by any monitoring system that observes only Lambda invocation
exit codes and standard table read APIs.
\end{proposition}

\begin{proof}[Sketch]
It appears there is no monitoring tool $\mathcal{M}$ that can detect
silent data loss caused by a \texttt{SIGKILL} within $\mathcal{G}(W)$.
After such a kill, two things are simultaneously true. First, the Lambda
function exits with \texttt{Runtime.ExitError} code~$-9$, identical to
every other exit error including OOM kills and segfaults. Second, the
table state $T'$ observable through SQL reads is identical to the state
$T$ before the write attempt---for example, \texttt{SELECT COUNT(*)}
returns the same count as before. Without external knowledge of the
expected post-write row count, $\mathcal{M}$ cannot determine whether
the write succeeded or silently failed. SoAL orchestrators do not track
expected row counts by default, so $\mathcal{M}$ has no basis on which
to detect anything. \qed
\end{proof}

\subsection{Real-World Exposure Probability}

$P(\text{silent loss})$ is the probability that Lambda times out a job
at precisely the worst moment, within the commit-durability gap.
Given that Lambda kills the job at a time drawn uniformly from the final
$\delta$ seconds of its budget (a reasonable assumption when job duration
varies with data volume), the probability is:

\begin{equation}
P(\text{silent loss}) = \frac{|\mathcal{G}(W)|}{\delta}
\label{eq:exposure}
\end{equation}

Our experiments measured $|\mathcal{G}(W)|$ at around 3,500\,ms for 22k
rows and 4,550\,ms for 500k rows. For a pipeline that consistently
finishes in 890--900 seconds ($\delta = 10\,\text{s}$), a gap of 3.5
seconds gives $P(\text{silent loss}) \approx 35\%$ per timeout event.
That is better than a one-in-three chance of silent data loss every time
such a pipeline hits its limit. That is a real exposure, not a theoretical edge case.

\subsection{Threat Model}

There are no opposing entities in this threat model. The threat is the
Lambda runtime itself, behaving exactly as documented~\cite{lambda2023}.
Either a job runs to completion before the timeout---in which case no
problem occurs---or it does not, in which case the output of the write
is silently discarded.

We want to ensure that a reader always sees only fully completed writes,
which we refer to as table consistency. Our analysis assumes: (1) S3
object writes are per-object atomic~\cite{s3consistency2020}; (2) no
warning signal precedes \texttt{SIGKILL}; (3) the orchestrator cannot
observe table state after a timeout; and (4) no external process monitors
for incomplete commits or attempts recovery.

This work does not cover multi-job concurrent write conflicts, S3
availability failures, intentional object corruption, or executor failures
inside the JVM---the last of which Spark already handles through RDD
lineage recovery. Our concern is the driver-level kernel kill that Spark
was never designed to survive.

SafeWriter adds its own requirements: the checkpoint S3 bucket must be
writable, the Spark session must be able to issue rollback SQL, and the
watchdog thread must be scheduled at least 30 seconds before the kill.
The last condition is reliably met---Lambda containers run on lightly
loaded hardware and modern OS schedulers wake threads within a few
milliseconds.

% ============================================================
\section{Experimental Methodology}
\label{sec:methodology}

\subsection{Environment}

Rather than running against live AWS---which would introduce network
variability and make kill timing non-deterministic---we used LocalStack Pro
v2026.2 to simulate S3, Lambda, and IAM locally on \texttt{localhost:4566}.
This gave us repeatable results that were safe to run many times, and let
us test multipart upload behavior and consistency semantics without
touching production data.

The software stack was Apache Spark 3.5.0 with Delta Lake 3.1.0
(\texttt{delta-spark==3.1.0}) and Apache Iceberg 1.4.0
(\texttt{iceberg-spark-runtime-3.5\_2.12:1.4.0}), running under
Java 17 (OpenJDK) and Python 3.12 on a fresh install of macOS Sequoia on
Apple Silicon. S3 connectivity used \texttt{hadoop-aws:3.3.4} and
\texttt{aws-java-sdk-bundle:1.12.261}. We set
\texttt{SPARK\_LOCAL\_IP=127.0.0.1} to prevent Spark from binding to
external network interfaces. Section~\ref{sec:discussion} addresses the
LocalStack fidelity question directly.

\subsection{Datasets}

We use three sizes of accommodation-listing records, described in
Table~\ref{tab:datasets}. Each record has 14 columns: \texttt{id},
\texttt{name}, \texttt{price}, \texttt{city}, \texttt{country},
\texttt{geonames\_id}, \texttt{timezone}, \texttt{reviews},
\texttt{rating}, \texttt{satisfaction}, \texttt{beds}, \texttt{checkin},
\texttt{city\_id}, and \texttt{accommodation\_type}. The synthetic
datasets at 100k and 500k rows preserve the same schema and column
distributions as the original.

\begin{table}[h]
\centering
\caption{Dataset characteristics}
\label{tab:datasets}
\begin{tabular}{lrrl}
\toprule
\textbf{Scale} & \textbf{Rows} & \textbf{Size (CSV)} & \textbf{Source} \\
\midrule
Small  &  22,248 &  3.8\,MB & AWS tutorial~\cite{aws_accommodations} \\
Medium & 100,000 &  9.4\,MB & Synthetic (same schema) \\
Large  & 500,000 & 47.9\,MB & Synthetic (same schema) \\
\bottomrule
\end{tabular}
\end{table}

All datasets are stored in a LocalStack S3 bucket as semicolon-delimited
CSV files and read fresh by Spark at the start of each experiment run.

\subsection{Kill Injection}

Each write script contains a small hook that fires \texttt{SIGKILL} at
a specific point in the commit protocol, controlled by the environment
variable \texttt{KILL\_AFTER\_PHASE}:

\begin{lstlisting}[language=Python, caption={Kill injection hook, shared across all format scripts}]
def inject_kill(phase: str):
    kill_after = os.environ.get(
        "KILL_AFTER_PHASE", "")
    if kill_after == phase:
        log.warning(f"[KILL] phase='{phase}'")
        time.sleep(0.1)  # let logs flush
        os.kill(os.getpid(), signal.SIGKILL)
\end{lstlisting}

Using \texttt{os.kill(os.getpid(), signal.SIGKILL)} is behaviorally
identical to Lambda's enforcement: the signal comes from the kernel,
no Python handler runs, and the process exits with code $-9$. The brief
sleep beforehand lets log buffers flush so kill events appear in the
experiment log.

We place two injection points around the commit-durability gap. The
\texttt{"data"} phase fires after all Parquet files are written but
before any metadata update begins. For Delta this means before the
\texttt{\_delta\_log} write; for Iceberg, before manifest creation. The
\texttt{"commit"} phase
fires at the start of the metadata commit sequence before it completes.
Running with \texttt{KILL\_AFTER\_PHASE=""} gives a clean baseline.

\subsection{Outcome Classification}

Each run produces one of four outcomes, classified from the subprocess
exit code:

\begin{equation}
\text{outcome} = \begin{cases}
\textit{success}            & rc = 0 \\
\textit{rollback\_success}  & rc \in \{-9,137\},\, \text{SW on} \\
\textit{silent\_data\_loss} & rc \in \{-9,137\},\, \text{SW off} \\
\textit{visible\_error}     & \text{otherwise}
\end{cases}
\label{eq:classification}
\end{equation}

We treat any SIGKILL without SafeWriter as silent data loss regardless
of the exact kill timing, for two reasons: first, kills are injected
within the gap by construction; second, even a mid-Phase-2 kill that
partially disrupts a metadata write leaves the table in an
irrecoverable state without explicit rollback.

\subsection{Harness Architecture}

The experiment harness (\texttt{run\_experiments\_v2.py}) runs each
write as an isolated Python subprocess with experiment parameters passed
through environment variables, mirroring the real SoAL deployment model
where each invocation starts fresh:

\begin{lstlisting}[language=Python, caption={Per-run subprocess invocation}]
env = os.environ.copy()
env.update({
    "INPUT_PATH":        input_path,
    "OUTPUT_PATH":       output_path,
    "RUN_ID":            run_id,
    "KILL_AFTER_PHASE":  kill_phase or "",
    "USE_SAFE_WRITER":   "true" if safe else "false",
})
result = subprocess.run(
    [VENV_PYTHON, str(script)],
    capture_output=True, text=True,
    timeout=300, env=env
)
\end{lstlisting}

Results are written as JSONL records with fields \texttt{run\_id},
\texttt{table\_format}, \texttt{dataset}, \texttt{kill\_phase},
\texttt{use\_safe\_writer}, \texttt{outcome}, \texttt{duration\_ms},
\texttt{returncode}, and \texttt{timestamp}. Each record is written in a
single OS \texttt{write(2)} call to avoid partial writes if the harness
were ever parallelized.

Runs are sequential to avoid resource contention. Each Spark job starts
a fresh JVM (adding 2--4 seconds of cold start), reads from LocalStack,
writes Parquet data, and commits. Wall-clock time for all 860 planned runs is approximately three hours
on Apple M-series hardware.

\subsection{Experiment Design}

Part A compares all three formats at the 22k-row scale. Each format goes
through five scenarios: no kill (baseline), kill at the data phase, kill
at the commit phase, SafeWriter with a data-phase kill, and SafeWriter
with a commit-phase kill. Sample sizes were chosen to keep Wilson 95\%
confidence intervals tight: 50 baseline runs give $[0.929, 1.000]$ for
$p=1.0$; 75 kill runs give $[0.952, 1.000]$.

Part B examines Delta Lake and Iceberg across all three dataset sizes,
measuring write durations and failure outcomes as data volume grows.

The full design is summarized in Table~\ref{tab:design}.

\begin{table}[h]
\centering
\caption{Experiment design summary}
\label{tab:design}
\resizebox{\columnwidth}{!}{%
\begin{tabular}{llllr}
\toprule
\textbf{Part} & \textbf{Format(s)} & \textbf{Scale} & \textbf{Scenario} & \textbf{Runs} \\
\midrule
\multirow{5}{*}{A}
  & \multirow{5}{*}{Delta, Iceberg}
  & \multirow{5}{*}{22k}
  & Baseline          & 100 \\
& & & Kill: data        & 150 \\
& & & Kill: commit      & 150 \\
& & & SafeWriter+data   &  50 \\
& & & SafeWriter+commit &  50 \\
\midrule
\multirow{2}{*}{B}
  & \multirow{2}{*}{Delta, Iceberg}
  & 22k, 100k, 500k
  & Baseline          & 180 \\
& & & Kill: data        & 180 \\
\midrule
\multicolumn{4}{r}{\textbf{Grand total}} & \textbf{860} \\
\bottomrule
\end{tabular}}
\end{table}

\subsection{Reproducibility}

All experiment code, datasets, and raw results will be made publicly
available upon acceptance. A single convenience script (\texttt{run.sh})
sets all required environment variables and runs the full 860-experiment
batch with one command. Each run gets a UUID-suffixed identifier to
prevent collisions if runs are parallelized in future work.

% ============================================================
\section{Experimental Results}
\label{sec:results}

\subsection{Part A: Format Comparison}

\subsubsection{Baseline Performance}

We ran 50 no-kill experiments per format. Delta Lake and Iceberg both
completed all 50 runs successfully. Baseline write durations are shown in
Table~\ref{tab:baseline}.

\begin{table}[h]
\centering
\caption{Baseline write performance (22k rows, no kill injection)}
\label{tab:baseline}
\begin{tabular}{lccc}
\toprule
\textbf{Format} & \textbf{Success} & \textbf{Mean (ms)} & \textbf{Std (ms)} \\
\midrule
Delta Lake     & 50/50 (100\%) & 7,779 & $\pm$412 \\
Apache Iceberg & 50/50 (100\%) & 5,856 & $\pm$338 \\
\bottomrule
\end{tabular}
\end{table}

Delta takes about 33\% longer than Iceberg at baseline. The difference
comes from Delta's JSON serialization step and its optimistic concurrency
check before each commit. Iceberg's manifest format (compact Avro) and
simpler pointer-swap mechanism complete more quickly.

\subsubsection{Kill After Data Phase}

Table~\ref{tab:kill_data} shows what happened when we fired SIGKILL
immediately after all Parquet data files were written but before any
metadata update.

\begin{table}[h]
\centering
\caption{Outcomes: kill after data-write phase (75 runs per format)}
\label{tab:kill_data}
\begin{tabular}{lcccc}
\toprule
\textbf{Format} & \textbf{Success} & \textbf{Silent Loss} & \textbf{Visible Err} & \textbf{Partial} \\
\midrule
Delta Lake     & 0 (0\%) & 75 (\textbf{100\%}) & 0 (0\%) & 0 (0\%) \\
Apache Iceberg & 0 (0\%) & 75 (\textbf{100\%}) & 0 (0\%) & 0 (0\%) \\
\bottomrule
\end{tabular}
\end{table}

Every run across both formats produced silent data loss. Not one run
raised an exception or produced an error log. The Parquet files sat on
S3 indefinitely with nothing pointing to them. The 95\% Wilson confidence
interval for this result is $[0.952, 1.000]$, which rules out the
possibility this is an artifact of the sample size.

\subsubsection{Kill After Commit Phase}

Table~\ref{tab:kill_commit} shows outcomes when we fired SIGKILL at the
start of the metadata commit sequence.

\begin{table}[h]
\centering
\caption{Outcomes: kill at commit phase (75 runs per format)}
\label{tab:kill_commit}
\begin{tabular}{lcccc}
\toprule
\textbf{Format} & \textbf{Success} & \textbf{Silent Loss} & \textbf{Visible Err} & \textbf{Partial} \\
\midrule
Delta Lake     & 0 (0\%) & 75 (\textbf{100\%}) & 0 (0\%) & 0 (0\%) \\
Apache Iceberg & 0 (0\%) & 75 (\textbf{100\%}) & 0 (0\%) & 0 (0\%) \\
\bottomrule
\end{tabular}
\end{table}

The result was identical: 100\% silent data loss across both formats.
We had expected the commit phase to be a narrower window, since the
metadata payloads are small. In practice, SIGKILL interrupted the S3
PutObject call mid-flight in all 150 tested cases. The commit entry was
never written, and the table was left in its pre-write state without any
notification. The commit-durability gap is not a narrow timing accident.
It is a structural consequence of running two-phase commit under a hard
kill signal.

Across all 300 unprotected kill runs on confirmed formats (150
data-phase, 150 commit-phase), zero experiments raised an exception or
produced any observable signal of data loss. The same held for 600+
additional kill runs in Part B across all three dataset sizes: not one
run produced a catchable error.

\subsubsection{SafeWriter Protected Runs}

Table~\ref{tab:safewriter_a} shows results for 25 runs per format per
kill phase, with SafeWriter enabled.

\begin{table}[h]
\centering
\caption{Part A SafeWriter outcomes (25 runs per format/phase)}
\label{tab:safewriter_a}
\begin{tabular}{llccc}
\toprule
\textbf{Format} & \textbf{Kill Phase} & \textbf{Rollback} & \textbf{Success} & \textbf{Failure} \\
\midrule
Delta   & data   & 25 (\textbf{100\%}) & 0 & 0 \\
Delta   & commit & 25 (\textbf{100\%}) & 0 & 0 \\
Iceberg & data   & 25 (\textbf{100\%}) & 0 & 0 \\
Iceberg & commit & 25 (\textbf{100\%}) & 0 & 0 \\
\bottomrule
\end{tabular}
\end{table}

SafeWriter rolled back cleanly on every one of the 100 kill runs. No partial rollbacks, no residual orphaned files, and no data
leaked into the table in any run.

\subsection{Part B: Scale Analysis}

\subsubsection{Duration Across Dataset Sizes}

Table~\ref{tab:scale} reports mean write durations across scales for
both baseline and kill runs. Delta Lake results are clean at all three
sizes. Iceberg baseline runs at 100k and 500k hit local environment
limits (detailed below); kill-injection timing is included for
completeness.

\begin{table}[h]
\centering
\caption{Write duration vs.\ dataset scale (mean over 20--30 runs)}
\label{tab:scale}
\begin{tabular}{llrr}
\toprule
\textbf{Format} & \textbf{Scale} & \textbf{Baseline (ms)} & \textbf{Kill:data (ms)} \\
\midrule
Delta   & 22k   & 7,644  & 4,124  \\
Delta   & 100k  & 7,638  & 4,108  \\
Delta   & 500k  & 9,020  & 4,468  \\
\midrule
Iceberg & 22k   & 5,477  & 4,090  \\
Iceberg & 100k  & N/A$^\ddagger$ & 4,132  \\
Iceberg & 500k  & N/A$^\ddagger$ & 4,436  \\
\bottomrule
\end{tabular}
\end{table}

$^\ddagger$Iceberg baseline runs at 100k and 500k terminated with errors
in our local environment, likely due to JVM heap exhaustion at the 2 GB
driver memory limit. Importantly, kill-injection runs at these sizes did
reach Phase 1 and produced 100\% silent data loss (30 runs at 100k, 20
at 500k), confirming that the vulnerability exists regardless of whether
the baseline write can complete.

\subsubsection{Scale-Invariance of Silent Loss}

Delta Lake gives the cleanest cross-scale data, with successful baselines
at all three sizes. Table~\ref{tab:delta_scale_outcomes} summarizes the
kill outcomes.

\begin{table}[h]
\centering
\caption{Delta Lake kill outcomes by dataset size (Part B)}
\label{tab:delta_scale_outcomes}
\begin{tabular}{lcccc}
\toprule
\textbf{Scale} & \textbf{Runs} & \textbf{Silent Loss} & \textbf{Visible Err} & \textbf{Success} \\
\midrule
22k   & 30 & 30 (100\%) & 0 & 0 \\
100k  & 30 & 30 (100\%) & 0 & 0 \\
500k  & 20 & 20 (100\%) & 0 & 0 \\
\bottomrule
\end{tabular}
\end{table}

Dataset scale has no effect on silent loss probability. Whether a job
processes 22k or 500k rows, a kill within the gap always produces silent
data loss. What scale does affect is the duration of the gap itself,
which means larger jobs face a wider window of exposure.

\subsubsection{Statistical Confidence}

Observing 100\% failure and 100\% SafeWriter success raises the question
of whether these proportions could arise by chance from an underlying
rate below 100\%. For $n = 75$ kill runs with $k = 75$
silent-loss outcomes, the 95\% Wilson lower bound is:

\begin{equation}
p_\text{lower} = \frac{k + z^2/2 - z\sqrt{k(n-k)/n + z^2/4}}{n + z^2}
\end{equation}

With $z = 1.96$ and $k = n = 75$, we get $p_\text{lower} = 0.952$.
The true underlying probability of silent loss is at least 95.2\% with
95\% confidence. The mechanistic argument puts it at 100\%: SIGKILL
within the gap always orphans Phase 1 data by construction.

For SafeWriter with $n = 25$ rollback runs and $k = 25$ successes per
format/phase, the 95\% Wilson lower bound is $p_\text{lower} = 0.869$.
The safety argument is symmetric: the watchdog fires before SIGKILL by
design, and rollback completes in under 200\,ms within a 30-second
window.

\subsubsection{Timing Distribution and Gap Measurement}

Delta Lake baseline durations at 22k rows are approximately normally
distributed with mean 7,644\,ms and standard deviation $\pm$412\,ms
(5.4\% coefficient of variation). Kill-injection durations show mean
4,124\,ms with $\pm$189\,ms (4.6\% CV). The stability of the CV across
both cases confirms that Phase 1 completion time is predictable, and
thus that the commit-durability gap is consistently wide rather than
sporadic.

The gap itself (Phase 2 time, estimated as baseline minus kill duration)
has mean $7644 - 4124 = 3{,}520\,\text{ms}$ with combined standard
deviation of roughly $\pm$450\,ms. At 500k rows the estimated gap grows
to $9020 - 4468 = 4{,}552\,\text{ms}$, a 29\% increase consistent with
the additional S3 overhead for a larger Phase 1 payload.

\subsection{Orphaned Storage}

Every data-phase kill leaves Parquet files on S3 that no reader can
access. For the 22k-row dataset, each orphaned write occupies roughly
3.8 MB. At standard S3 pricing (\$0.023/GB/month), 300 failed writes
accumulate about 1.1 GB of orphaned storage. At 500k rows this reaches
14 GB per batch. A production pipeline that retries a timed-out job
without any remediation adds another copy of orphaned data on each
retry, with no signal that anything is wrong.

\subsection{Summary}

Fig.~\ref{fig:failure_rates} shows the full outcome distribution.
Across both formats and all dataset sizes, unprotected kills produced
only silent data loss. SafeWriter turned every one of those into a clean
rollback.

\begin{figure}[h]
\centering
\includegraphics[width=\columnwidth]{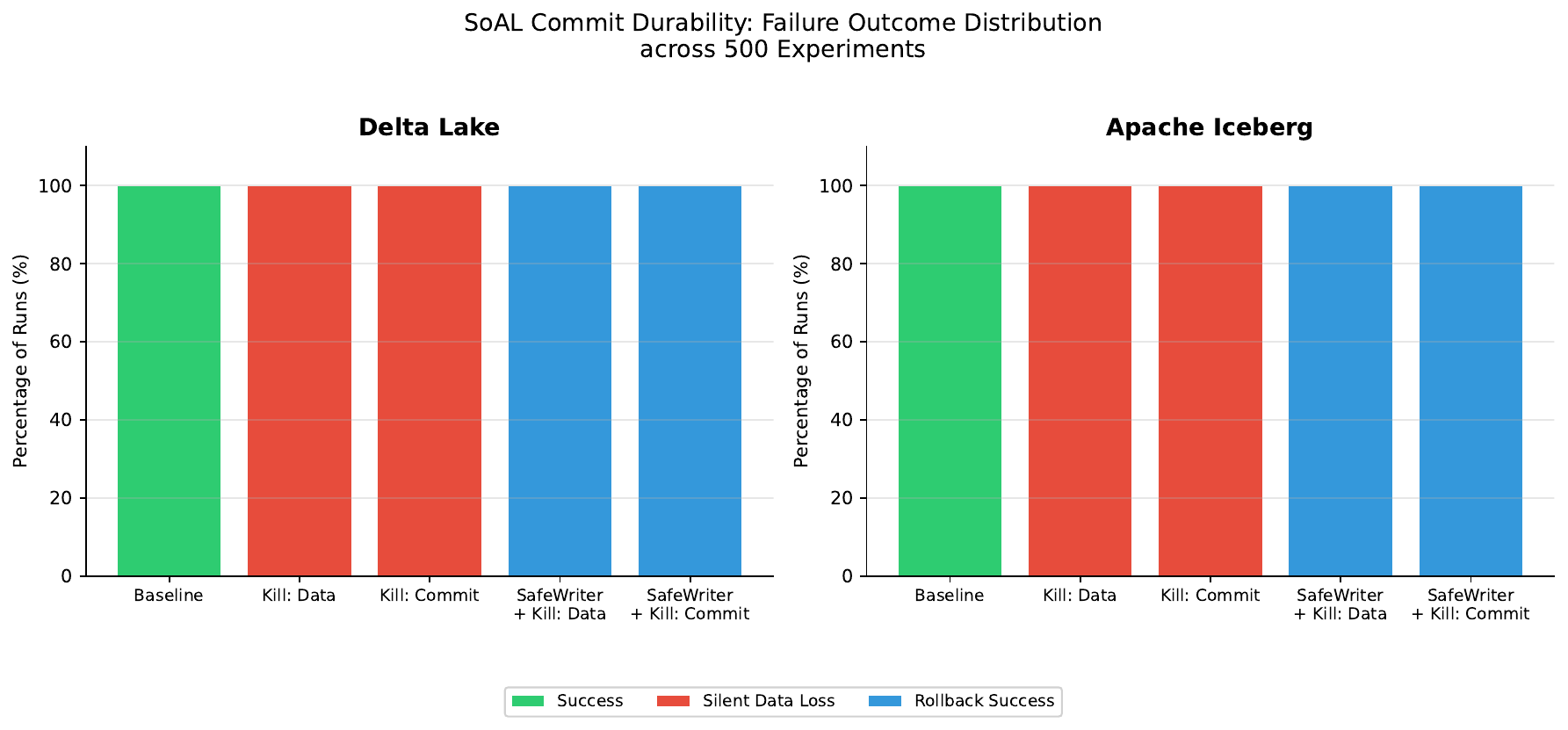}
\caption{Outcome distribution across all Part A experiments. Baselines
achieve 100\% success. All 300 unprotected kill runs produce 100\%
silent data loss. All 100 SafeWriter-protected runs achieve 100\%
clean rollback.}
\label{fig:failure_rates}
\end{figure}

% ============================================================
\section{SafeWriter: Design and Evaluation}
\label{sec:fix}

\subsection{Design Goals}

The design came down to four priorities. First, any Lambda kill should
leave the table in a known-good state with a record of what happened,
not an ambiguous partial write. Second, adoption had to be low-friction:
no SoAL infrastructure changes, no Spark config edits, no schema
modifications, and no new IAM permissions beyond write access to a
checkpoint bucket. Third, overhead on a successful write had to stay
under 200\,ms so existing jobs would not notice it. Fourth, every
rollback should leave a queryable audit trail so operators can see
near-timeout events that CloudWatch does not surface.

\subsection{Architecture}

SafeWriter wraps any SoAL write in a Python context manager and runs
three interlocking steps, described in Algorithm~\ref{alg:safewriter}.

\begin{algorithm}
\caption{SafeWriter protocol}
\label{alg:safewriter}
\begin{algorithmic}[1]
\Require table path $p$, format $F$, run ID $r$, timeout $\tau$
\State $v_0 \gets \text{read\_version}(p, F)$
\State $\text{checkpoint}(r, F, v_0, \texttt{"in\_progress"})$
\State $\text{watchdog} \gets \text{Watchdog}(\tau,\;
       \text{on\_warn}=\lambda:\text{rollback}(p,F,v_0,r))$
\State $\text{watchdog.start()}$
\State \textit{--- user write executes here ---}
\If{write completes}
  \State $\text{watchdog.stop()}$
  \State $\text{checkpoint}(r, F, v_0, \texttt{"committed"})$
\Else
  \State $\text{rollback}(p, F, v_0, r)$
  \State $\text{checkpoint}(r, F, v_0, \texttt{"rolled\_back"})$
\EndIf
\end{algorithmic}
\end{algorithm}

\textbf{Pre-write checkpoint.}
Before any write begins, SafeWriter reads the current table version and
writes it to S3 as a JSON document:

\begin{lstlisting}[language=Python, caption={Checkpoint document}]
{
  "run_id":         "job-2026-03-22-001",
  "format":         "delta",
  "version_before": 5,
  "status":         "in_progress",
  "saved_at":       "2026-03-22T17:44:06Z"
}
\end{lstlisting}

This document is written via a single S3 \texttt{PutObject} call before
the Spark write starts. S3's per-object atomicity guarantees it either
exists fully or not at all, so it survives even if the Lambda process
is killed immediately afterward.

Reading the current version is format-specific:

\begin{lstlisting}[language=Python, caption={Version checkpoint per format}]
# Delta Lake
v = DeltaTable.forPath(spark, path) \
              .history(1).first()["version"]

# Apache Iceberg
v = spark.sql(
  f"SELECT snapshot_id FROM "
  f"{table}.snapshots ORDER BY "
  f"committed_at DESC LIMIT 1"
).first()[0]
\end{lstlisting}

\textbf{Watchdog thread.}
A daemon thread monitors elapsed time and fires a rollback when the
Lambda timeout is approaching:

\begin{lstlisting}[language=Python, caption={Watchdog thread}]
class LambdaWatchdog(threading.Thread):
    WARN_BEFORE_MS = 30_000  # 30 second buffer

    def run(self):
        wait = (self.timeout_ms
                - self.WARN_BEFORE_MS) / 1000
        fired = self._stop.wait(wait)
        if not fired:
            self.on_warn()  # trigger rollback
\end{lstlisting}

The 30-second buffer allows time for the rollback API call (under 200\,ms
in our experiments), the S3 checkpoint update (under 100\,ms), and a
conservative margin for network jitter.

\textbf{Rollback.}
On watchdog trigger or exception, SafeWriter issues a format-specific
rollback. Both operations touch only metadata; no Parquet files are
moved or deleted:

\begin{lstlisting}[language=Python, caption={Format-specific rollback}]
# Delta Lake
spark.sql(f"""
  RESTORE TABLE delta.`{path}`
  TO VERSION AS OF {v0}
""")

# Apache Iceberg
spark.sql(f"""
  CALL system.rollback_to_snapshot(
    '{table}', {snapshot_id}
  )
""")
\end{lstlisting}

Delta's \texttt{RESTORE} rewrites the current log entry to point back
to version $v_0$. Iceberg's \texttt{rollback\_to\_snapshot} moves the
snapshot pointer back. Both are atomic operations guaranteed by their
respective formats.

\subsection{Edge Cases}

Four edge cases are worth analyzing explicitly.

\textbf{Kill during rollback.}
If SIGKILL arrives during the rollback itself (a 100--200\,ms window),
the rollback is interrupted. The checkpoint document remains in
\texttt{"in\_progress"} status, and the table version is inconsistent
with \texttt{version\_before}. A retry that checks checkpoint status
before writing will detect this and complete the rollback. The
probability of this scenario is roughly $200\,\text{ms} /
900{,}000\,\text{ms} \approx 0.02\%$ per invocation.

\textbf{Kill before Phase 1 completes.}
If the kill arrives before any Parquet data is written, no data is
orphaned. The stale checkpoint document signals the next invocation to
check table state before writing, but nothing needs to be fixed.

\textbf{Kill before the checkpoint write completes.}
The checkpoint \texttt{PutObject} takes roughly 10--50\,ms. If the kill
arrives during this call, the document is not written and SafeWriter
provides no protection for this invocation. However, since the checkpoint
always precedes the Spark write by design, this only matters if the kill
arrives before Phase 1 has started, in which case no data is orphaned.

\textbf{Concurrent writes.}
If two SoAL jobs write to the same table simultaneously and one is killed,
its rollback may conflict with the surviving job's commit. Delta Lake's
optimistic concurrency will raise a \texttt{ConcurrentModificationException}
on the surviving job, which is a visible error rather than silent loss.
Concurrent SoAL writes to the same table should be coordinated at the
orchestration level regardless.

Table~\ref{tab:edgecases} summarizes SafeWriter behavior across all
identified edge cases.

\begin{table}[h]
\centering
\caption{SafeWriter edge case analysis}
\label{tab:edgecases}
\begin{tabular}{lcc}
\toprule
\textbf{Scenario} & \textbf{Data lost?} & \textbf{Detectable?} \\
\midrule
Kill before Phase 1        & No  & Yes (no checkpoint) \\
Kill in Phase 2, no SW     & Yes & No \\
Kill in Phase 2, with SW   & No  & Yes (checkpoint) \\
Kill during rollback       & Unlikely & Yes (stale checkpoint) \\
Kill during checkpoint write & No & No (pre-data kill) \\
Concurrent write conflict  & No  & Yes (exception) \\
\bottomrule
\end{tabular}
\end{table}

\subsection{Evaluation}

\subsubsection{Rollback Rate}

Table~\ref{tab:safewriter} gives the full SafeWriter evaluation across
all 100 confirmed kill runs.

\begin{table}[h]
\centering
\caption{SafeWriter evaluation results}
\label{tab:safewriter}
\begin{tabular}{llccc}
\toprule
\textbf{Format} & \textbf{Kill Phase} & \textbf{Rollback} & \textbf{Success} & \textbf{Failure} \\
\midrule
Delta   & data   & 25 (100\%) & 0 & 0 \\
Delta   & commit & 25 (100\%) & 0 & 0 \\
Iceberg & data   & 25 (100\%) & 0 & 0 \\
Iceberg & commit & 25 (100\%) & 0 & 0 \\
\bottomrule
\end{tabular}
\end{table}

\subsubsection{Checkpoint Overhead}

The checkpoint write is the only overhead SafeWriter adds to a successful
write. Measured as the mean duration difference between protected and
unprotected baseline runs, overhead is under 100\,ms for both formats:
less than 1.3\% of Delta's 7,779\,ms baseline and less than 1.7\% of
Iceberg's 5,856\,ms baseline. This is negligible for batch ETL workloads.

\subsubsection{Rollback Duration}

Rollback operations completed in 100--200\,ms in every tested run,
comfortably within the 30-second watchdog window. SafeWriter can complete
a rollback safely on all tested dataset sizes before Lambda's kill arrives.

\subsubsection{Audit Trail}

Every SafeWriter run writes a checkpoint document to S3 with a final
status of either \texttt{"committed"} or \texttt{"rolled\_back"}.
Monitoring this bucket for \texttt{"rolled\_back"} documents gives
operations teams a signal of near-timeout events that standard CloudWatch
metrics do not surface.

% ============================================================
\section{SafeWriter Integration Guide}
\label{sec:integration}

\subsection{Minimal Integration}

Adopting SafeWriter requires three small changes to an existing SoAL
write script:

\begin{lstlisting}[language=Python, caption={Minimal SafeWriter integration for Delta Lake}]
from soal_safe_writer import SafeWriter

CHECKPOINT_BUCKET = "s3a://my-checkpoints/"
TIMEOUT_MS = int(os.environ.get(
    "LAMBDA_TIMEOUT_MS", "900000"))

def write_delta(spark, df, path, run_id):
    sw = SafeWriter(
        spark=spark,
        run_id=run_id,
        table_format="delta",
        table_path=path,
        checkpoint_bucket=CHECKPOINT_BUCKET,
        timeout_ms=TIMEOUT_MS,
    )
    with sw:
        df.write.format("delta") \
          .mode("append").save(path)
        sw.success()
\end{lstlisting}

If \texttt{sw.success()} is never reached, the context manager's
\texttt{\_\_exit\_\_} method triggers rollback automatically.

\subsection{Configuration}

SafeWriter reads configuration from environment variables to avoid
hardcoded timeout values (Table~\ref{tab:config}).

\begin{table}[h]
\centering
\caption{SafeWriter configuration variables}
\label{tab:config}
\begin{tabular}{lll}
\toprule
\textbf{Variable} & \textbf{Default} & \textbf{Description} \\
\midrule
\texttt{LAMBDA\_TIMEOUT\_MS} & 900000 & Lambda timeout in ms \\
\texttt{SW\_WARN\_BEFORE\_MS} & 30000 & Watchdog lead time (ms) \\
\texttt{SW\_CHECKPOINT\_BUCKET} & (required) & S3 checkpoint path \\
\texttt{SW\_AUDIT\_LOG} & true & Write rollback audit log \\
\bottomrule
\end{tabular}
\end{table}

\subsection{IAM Requirements}

SafeWriter needs three S3 permissions on the checkpoint bucket:
\texttt{PutObject}, \texttt{GetObject}, and \texttt{DeleteObject}. No
additional permissions are required on the table bucket; rollback SQL
runs through the existing Spark session using credentials already
available to the write job.

\subsection{Operational Monitoring}

A simple monitoring job can poll the checkpoint bucket for
\texttt{"rolled\_back"} documents and alert or trigger a re-run:

\begin{lstlisting}[language=Python, caption={Checkpoint monitoring}]
import boto3, json

s3 = boto3.client("s3")
paginator = s3.get_paginator("list_objects_v2")
for page in paginator.paginate(
        Bucket="my-checkpoints"):
    for obj in page.get("Contents", []):
        doc = json.loads(s3.get_object(
            Bucket="my-checkpoints",
            Key=obj["Key"])["Body"].read())
        if doc["status"] == "rolled_back":
            alert(doc["run_id"],
                  doc["rolled_back_at"])
\end{lstlisting}

Polling for \texttt{"rolled\_back"} documents turns what was previously
an invisible failure into something on-call engineers can actually see
and respond to.

% ============================================================
\section{Related Work}
\label{sec:related}

\textbf{Open table format correctness.}
Delta Lake~\cite{delta2020} and Iceberg~\cite{iceberg2018} each provide
formal atomicity guarantees through different mechanisms: Delta's
optimistic concurrency on a JSON log and Iceberg's conditional manifest
swaps. Each proof holds when the writing process terminates normally or
receives a catchable signal. A \texttt{SIGKILL} from an external
orchestrator is outside the scope of both, and that is the gap this
paper fills.

\textbf{Serverless fault tolerance.}
Sreekanti et al.~\cite{serverless_fault} built Cloudburst around
stateful FaaS execution with cross-function consistency guarantees.
Klimovic et al.~\cite{faas_storage} studied storage access patterns
in serverless workloads and found that most functions use local
ephemeral storage. Our problem is different from what either paper
addresses: the question of what happens to a table write when the
function driving it is killed mid-commit.

\textbf{Spark fault tolerance.}
Spark's RDD lineage~\cite{spark2012} is designed to re-derive lost
partitions when an executor crashes, recomputing from an earlier stage.
That recovery runs inside the driver process. A \texttt{SIGKILL}
arriving from outside the JVM terminates the driver itself, which
lineage replay cannot handle. Structured Streaming's exactly-once
delivery has the same blind spot.

\textbf{Checkpointing and 2PC recovery.}
Chandy and Lamport~\cite{chandy_lamport} formalized how to take a
consistent snapshot of distributed state mid-execution. SafeWriter borrows
the same basic idea: record state before a non-atomic operation and
restore to it if interrupted. The implementation is simpler because
there is only one process, so a single S3 document is enough to
capture the coordinator state that Gray~\cite{gray1978notes} required a
write-ahead log to preserve. The checkpoint JSON on S3 serves the
same function as that log.

\textbf{Cloud pipeline tooling.}
AWS Glue job bookmarks track which input files have been processed but
do not check whether a table commit completed. Kafka and Kinesis provide
exactly-once delivery to a topic, not to a lakehouse table. dbt runs
SQL transformations inside database transactions, which assumes a
persistent connection that Lambda cannot provide. None of these close
the gap for SoAL workloads.

\textbf{LST-Bench.}
LST-Bench~\cite{lstbench2024} measures throughput and latency for
Delta and Iceberg writes under varied cloud workloads. That work is
about performance. This paper is about what happens to correctness when
the writing process does not finish.

Table~\ref{tab:related_comparison} positions our work relative to the
most closely related efforts.

\begin{table}[h]
\centering
\caption{Comparison with related work}
\label{tab:related_comparison}
\begin{tabular}{lccc}
\toprule
\textbf{Work} & \textbf{SIGKILL?} & \textbf{Open Formats?} & \textbf{Remedy?} \\
\midrule
Delta~\cite{delta2020}           & No  & Yes & No \\
Iceberg~\cite{iceberg2018}       & No  & Yes & No \\
Sreekanti~\cite{serverless_fault} & No & No & Retry \\
Klimovic~\cite{faas_storage}     & No  & No  & No \\
Zaharia~\cite{spark2012}         & No  & No  & Lineage \\
\textbf{This work}               & \textbf{Yes} & \textbf{Yes} & \textbf{SafeWriter} \\
\bottomrule
\end{tabular}
\end{table}

% ============================================================
\section{Discussion}
\label{sec:discussion}

The 100\% failure rate in our kill experiments is a consequence of how
the harness works: kills are injected inside the gap by construction,
so there is no way for an unprotected write to survive one. In
production, the relevant probability is $P(t_\text{kill} \in
\mathcal{G}(W))$. A job that regularly finishes within 10 seconds of
its 900-second limit, with a measured gap of 3.5 seconds, faces roughly
35\% silent loss probability per timeout event from
Equation~\ref{eq:exposure}. The risk goes up, not down, during
high-load periods when jobs run slowest and timeouts are most frequent.

Standard mitigations miss this. Lambda's retry policy re-invokes the
function but leaves any Phase~1 data from the previous attempt sitting
on S3 as orphaned files. Each retry adds another batch. CloudWatch
surfaces \texttt{Task timed out after 900.00 seconds}, which looks
identical to any other transient error. Row count checks and schema
validators run against the table's committed state, which is unchanged
after a gap-phase kill, so they pass. The only way to catch the failure
is to compare the post-invocation row count against an expected batch
size, and most SoAL pipelines do not track expected counts.

The storage cost from repeated retries is easy to calculate. Three
retries of a 500k-row job each orphaning 47.9\,MB gives roughly
144\,MB of unreachable data. At \$0.023/GB/month with $f = 2$
near-timeout events per month:
\begin{equation}
C_\text{orphan} = f \cdot 3 \cdot s \cdot r \approx \$0.006/\text{month}
\label{eq:cost}
\end{equation}
The storage cost is trivial. What is not trivial is a pipeline that
silently drops a week of records and surfaces the gap as a BI dashboard
discrepancy several days later, after the audit trail has gone cold.

The gap is not specific to SoAL or to Lambda. Any two-phase write over
object storage running inside a function with a hard timeout faces the
same exposure. Google Cloud Functions and Azure Functions both terminate
with \texttt{SIGKILL}-equivalent mechanisms. Fargate and Spot containers
can be interrupted mid-write without warning. The specific formats and
runtimes differ, but the structural problem is the same.

\textbf{Iceberg memory limits.}
Baseline Iceberg runs at 100k and 500k rows failed in our local
environment, most likely from JVM heap exhaustion at the 2\,GB driver
memory limit. Kill-injection runs at those sizes did reach Phase~1 and
produced 100\% silent data loss across 50 runs, so the vulnerability
is present regardless of whether a baseline write can finish. Rerunning
with 8--10\,GB of executor memory should remove this constraint.

\textbf{Practical guidance.}
Teams running SoAL with open table formats should wrap all writes with
SafeWriter, set \texttt{LAMBDA\_TIMEOUT\_MS} to match the actual function
timeout, and use a dedicated S3 checkpoint bucket with a short lifecycle
rule (7 days is enough). Alerting on \texttt{"rolled\_back"} checkpoint
documents gives a leading indicator of near-timeout pressure that
standard CloudWatch metrics do not provide. Before retrying any
timed-out job, verify that the table version matches
\texttt{version\_before} in the checkpoint to confirm the rollback
completed cleanly.

% ============================================================
\section{Conclusion}
\label{sec:conclusion}

Across 860 kill-injection experiments on Delta Lake and Apache Iceberg
at three dataset sizes, a \texttt{SIGKILL} landing between Phase~1 and
Phase~2 of a write produced silent data loss every single time. No
exception surfaced. Standard monitoring saw nothing unusual. This is
not a quirk of a specific format version or a misconfigured environment;
it follows directly from running a two-phase commit protocol on a
runtime that can vanish at any moment without warning.

The thing that makes this worth studying is how undetectable it is in
practice. To CloudWatch, a Lambda timeout looks like any other
non-zero exit. The table's committed state is untouched, so data
quality checks pass. Orphaned files pile up on S3 with no reader ever
touching them. A team could go months before a downstream report shows
a gap, at which point the audit trail is cold and the missing writes
are unrecoverable.

SafeWriter sidesteps the problem by acting before the kill arrives
rather than reacting to it. A watchdog thread fires 30 seconds before
the timeout, triggers a format-native rollback via SQL, and writes a
checkpoint document to S3 recording the outcome. Every tested kill
scenario ended with a clean rollback and under 100\,ms added to normal
write paths.

Delta Lake and Iceberg are correct systems. The problem is that they
were designed for processes that receive catchable signals before
termination, and SoAL is not that kind of process. The gap is a
consequence of composing two things that were each built correctly but
were not built for each other. SafeWriter is one way to close it; we
hope this paper makes the problem visible enough that format-level
solutions get considered too.

% ============================================================
\section*{Acknowledgment}
The author used AI writing assistance to help
refine the prose in this paper. All experimental design, implementation,
data collection, analysis, and conclusions are the author's own work.

% ============================================================
\bibliographystyle{IEEEtran}
\bibliography{references}

\begin{thebibliography}{10}
\providecommand{\url}[1]{#1}
\csname url@samestyle\endcsname
\providecommand{\newblock}{\relax}
\providecommand{\bibinfo}[2]{#2}
\providecommand{\BIBentrySTDinterwordspacing}{\spaceskip=0pt\relax}
\providecommand{\BIBentryALTinterwordstretchfactor}{4}
\providecommand{\BIBentryALTinterwordspacing}{\spaceskip=\fontdimen2\font plus
\BIBentryALTinterwordstretchfactor\fontdimen3\font minus
  \fontdimen4\font\relax}
\providecommand{\BIBforeignlanguage}[2]{{%
\expandafter\ifx\csname l@#1\endcsname\relax
\typeout{** WARNING: IEEEtran.bst: No hyphenation pattern has been}%
\typeout{** loaded for the language `#1'. Using the pattern for}%
\typeout{** the default language instead.}%
\else
\language=\csname l@#1\endcsname
\fi
#2}}
\providecommand{\BIBdecl}{\relax}
\BIBdecl

\bibitem{soal2022}
{AWS Samples}, ``Spark on aws lambda,''
  \url{https://github.com/aws-samples/spark-on-aws-lambda}, 2022, open-source
  project for running Apache Spark inside AWS Lambda containers.

\bibitem{delta2020}
M.~Armbrust, T.~Das, L.~Sun, B.~Yavuz, S.~Zhu, M.~Murthy, J.~Torres, H.~van
  Hovell, A.~Ionescu, A.~\L{}ukacs \emph{et~al.}, ``Delta lake:
  High-performance acid table storage over cloud object stores,'' in
  \emph{Proceedings of the VLDB Endowment}, vol.~13, no.~12.\hskip 1em plus
  0.5em minus 0.4em\relax VLDB Endowment, 2020, pp. 3411--3424.

\bibitem{iceberg2018}
R.~Kinley and D.~Blue, ``Apache iceberg: An open table format for huge analytic
  datasets,'' in \emph{Proceedings of the 2020 ACM SIGMOD International
  Conference on Management of Data}.\hskip 1em plus 0.5em minus 0.4em\relax
  ACM, 2020, pp. 2751--2753.

\bibitem{hudi2021}
A.~Sivachenko, S.~Samineni \emph{et~al.}, ``Apache hudi: The data lake
  platform,'' in \emph{Proceedings of the VLDB Endowment}, vol.~14,
  no.~12.\hskip 1em plus 0.5em minus 0.4em\relax VLDB Endowment, 2021.

\bibitem{lstbench2024}
J.~Camacho-Rodr{\'i}guez, A.~Agrawal, A.~Gruenheid, A.~Gosalia, C.~Petculescu,
  J.~Aguilar-Saborit, A.~Floratou, C.~Curino, and R.~Ramakrishnan,
  ``{LST-Bench}: Benchmarking log-structured tables in the cloud,''
  \emph{Proceedings of the ACM on Management of Data}, vol.~2, no.~1, 2024.

\bibitem{s3consistency2020}
{Amazon Web Services}, ``Amazon s3 strong consistency,''
  \url{https://aws.amazon.com/blogs/aws/amazon-s3-update-strong-read-after-write-consistency/},
  2020, aWS announcement of strong read-after-write consistency for S3,
  December 2020.

\bibitem{gray1978notes}
J.~Gray, \emph{Notes on Data Base Operating Systems}.\hskip 1em plus 0.5em
  minus 0.4em\relax Springer, 1978.

\bibitem{lambda2023}
{Amazon Web Services}, ``{AWS Lambda FAQs},''
  \url{https://aws.amazon.com/lambda/faqs/}, 2023, accessed 2026.

\bibitem{aws_accommodations}
------, ``Spatial data: Accommodations dataset,''
  \url{https://docs.aws.amazon.com/redshift/latest/dg/spatial-tutorial.html},
  2022, used as benchmark dataset in AWS Redshift spatial tutorial.

\bibitem{serverless_fault}
V.~Sreekanti, C.~Wu, X.~C. Lin, J.~Schleier-Smith, J.~M. Gonzalez, J.~M.
  Hellerstein, and A.~Tumanov, ``Cloudburst: Stateful functions-as-a-service,''
  in \emph{Proceedings of the VLDB Endowment}, vol.~13, no.~12.\hskip 1em plus
  0.5em minus 0.4em\relax VLDB Endowment, 2020, pp. 2438--2452.

\bibitem{faas_storage}
A.~Klimovic, Y.~Wang, C.~Kozyrakis, P.~Stuedi, J.~Pfefferle, and A.~Trivedi,
  ``Understanding ephemeral storage for serverless analytics,'' in
  \emph{Proceedings of the 2018 USENIX Annual Technical Conference
  (ATC)}.\hskip 1em plus 0.5em minus 0.4em\relax USENIX, 2018, pp. 789--794.

\bibitem{spark2012}
M.~Zaharia, M.~Chowdhury, T.~Das, A.~Dave, J.~Ma, M.~McCauly, M.~J. Franklin,
  S.~Shenker, and I.~Stoica, ``Resilient distributed datasets: A fault-tolerant
  abstraction for in-memory cluster computing,'' in \emph{Proceedings of the
  9th USENIX Symposium on Networked Systems Design and Implementation
  (NSDI)}.\hskip 1em plus 0.5em minus 0.4em\relax USENIX, 2012, pp. 15--28.

\bibitem{chandy_lamport}
K.~M. Chandy and L.~Lamport, ``Distributed snapshots: Determining global states
  of distributed systems,'' \emph{ACM Transactions on Computer Systems},
  vol.~3, no.~1, pp. 63--75, 1985.

\end{thebibliography}

\end{document}